\documentclass[acmsmall]{acmart}

\def\BibTeX{{\rm B\kern-.05em{\sc i\kern-.025em b}\kern-.08emT\kern-.1667em\lower.7ex\hbox{E}\kern-.125emX}}
    
\newcommand{\rb}[2]{\raisebox{#1 mm}[0mm][0mm]{#2}}
\newcommand{\istrut}[2][0]{\rule[- #1 mm]{0mm}{#1 mm}\rule{0mm}{#2 mm}}

\newcommand{\hcm}[1][1]{\hspace*{#1 cm}}
\newcommand{\dist}{\operatorname{dist}}
\newcommand{\poly}{\operatorname{poly}}

\newtheorem*{remark}{Remark}

\setcopyright{acmcopyright}
\acmJournal{TALG}

\begin{document}

%
\title{Lower Bounds on Sparse Spanners, Emulators, and Diameter-reducing Shortcuts}

%
\author{Shang-En Huang}
\email{sehuang@umich.edu}
\affiliation{%
  \institution{University of Michigan}
}
\author{Seth Pettie}
\email{pettie@umich.edu}
\affiliation{%
  \institution{University of Michigan}
}

%

%
\begin{abstract}
We prove better lower bounds on additive spanners and emulators,
which are lossy compression schemes for \emph{undirected} graphs, 
as well as lower bounds on \emph{shortcut sets}, which reduce the diameter of \emph{directed} graphs.
We prove that 
any $O(n)$-size shortcut set cannot bring the diameter
below $\Omega(n^{1/6})$, 
and that any $O(m)$-size 
shortcut set cannot bring it below $\Omega(n^{1/11})$.  
These improve Hesse's~\cite{Hesse03} lower bound of $\Omega(n^{1/17})$.
By combining these constructions with Abboud and Bodwin's~\cite{abboud_bodwin_2017} edge-splitting technique, we get additive stretch lower bounds of $+\Omega(n^{1/11})$ for $O(n)$-size spanners and $+\Omega(n^{1/18})$ for $O(n)$-size emulators.  
These improve Abboud and Bodwin's $+\Omega(n^{1/22})$ lower bounds for both spanners and emulators.
\end{abstract}

%
%


%

%

%
\maketitle

\section{Introduction}

A {\em spanner} of an undirected unweighted graph $G=(V,E)$ is a subgraph $H$ that approximates the distance function of $G$ up to some \emph{stretch}.  An {\em emulator} for $G$ is defined similarly, except that $H$ need not be a subgraph, and may contain \emph{weighted} edges.  In this paper we consider only \emph{additive} stretch functions:
\[
\dist_G(u,v) \leq \dist_H(u,v) \leq \dist_G(u,v) + \beta,
\]
where $\beta$ may depend on $n$.

Graph compression schemes (like spanners and emulators) are related to the problem of \emph{shortcutting} digraphs to reduce diameter, inasmuch as lower bounds for both objects are constructed using the same suite of techniques.  These lower bounds begin from the construction of graphs 
in which numerous pairs of vertices have shortest paths that are \emph{unique}, \emph{edge-disjoint}, and relatively \emph{long}.  
Such graphs were independently discovered by Alon~\cite{Alon01}, Hesse~\cite{Hesse03}, and Coppersmith and Elkin~\cite{CoppersmithE06}; see also~\cite{abboud_bodwin_2017,AbboudBP17}.  Given such a ``base graph,''
derived graphs can be obtained through a variety of graph products such as the {\em alternation} product discovered independently by Hesse~\cite{Hesse03} and Abboud and Bodwin~\cite{abboud_bodwin_2017}
and the \emph{substitution} product used by Abboud and Bodwin~\cite{abboud_bodwin_2017} and developed further by 
Abboud, Bodwin, and Pettie~\cite{AbboudBP17}.

In this paper we apply the techniques developed in~\cite{Alon01,Hesse03,CoppersmithE06,abboud_bodwin_2017,AbboudBP17} to
obtain better lower bounds on shortcutting sets, 
additive spanners, and additive emulators.

\paragraph*{Shortcutting Sets.} Let $G=(V,E)$ be a directed graph
and $G^*=(V,E^*)$ its transitive closure.  The \emph{diameter}
of a digraph $G$ is the maximum of $\dist_G(u,v)$ 
over all pairs $(u,v)\in E^*$.  Thorup~\cite{thorup1992shortcutting} conjectured that it 
is possible to reduce the diameter of any digraph to $\poly(\log n)$
by adding a set $E'\subseteq E^*$ of at most $m = |E|$ \emph{shortcuts},
i.e., $G'=(V,E\cup E')$ would have diameter $\poly(\log n)$.
This conjecture was confirmed for a couple special graph classes~\cite{thorup1992shortcutting,thorup1995shortcutting}, but refuted in general by Hesse~\cite{Hesse03}, who exhibited a 
graph with $m=\Theta(n^{19/17})$ edges and diameter $\Theta(n^{1/17})$ 
such that any diameter-reducing shortcutting requires 
$\Omega(mn^{1/17})$ shortcuts.  More generally, there exist graphs with $m=n^{1+\epsilon}$ edges and diameter $n^\delta$, $\delta=\delta(\epsilon)$, that require $\Omega(n^{2-\epsilon})$ shortcuts to make the diameter $o(n^{\delta})$; see Abboud, Bodwin, and Pettie~\cite[\S 6]{AbboudBP17} for an alternative proof of this result.

On the upper bound side, it is trivial to reduce the diameter
to $\tilde{O}(\sqrt{n})$ with $O(n)$ shortcuts or diameter $\tilde{O}(n/\sqrt{m})$ with $O(m)$ shortcuts.\footnote{Pick a set $S$ of $\sqrt{n}$ or $\sqrt{m}$ vertices uniformly at random, and include $S^2 \cap E^*$ as shortcuts.} Unfortunately, the trivial shortcutting schemes are not efficiently constructible in near-linear time.  In some applications of shortcuttings, efficiency of the construction is just as important as reducing the diameter.  
For example, a longstanding problem in parallel computing is to 
\emph{simultaneously} achieve time and work efficiency in computing reachability.\footnote{This is the notorious \emph{transitive closure bottleneck}.}  Very recently, 
Fineman~\cite{Fineman18} proved that an $\tilde{O}(n)$-size shortcut set can be 
computed in near-optimal work $\tilde{O}(m)$ (and $\tilde{O}(n^{2/3})$ parallel time) 
that reduces the diameter to $\tilde{O}(n^{2/3})$.

In this paper we prove that $O(n)$-size shortcut sets cannot reduce
the diameter below $\Omega(n^{1/6})$, and that $O(m)$-size shortcut sets cannot reduce it below $\Omega(n^{1/11})$.  See Table~\ref{tab:shortcut}.

\begin{table}
\centering
\begin{tabular}{l@{\hcm[.3]}|l@{\hcm[.3]}|@{\hcm[.3]}l@{\hcm[.3]}|@{\hcm[.3]}l}
\multicolumn{1}{l}{\bf Citation} &
\multicolumn{1}{@{\hcm[0]}l}{\bf Shortcut Set Size\ \ \ } &
\multicolumn{1}{@{\hcm[-.1]}l}{\bf Diameter} &
\multicolumn{1}{@{\hcm[-.1]}l}{\bf Computation Time}  \\\hline\hline
\rb{-3}{Folklore/trivial}    &   $O(n)$              & $\tilde{O}(\sqrt{n})$     & $O(m\sqrt{n})$\istrut[1.5]{4.5}\\\cline{2-4}
                    &   $O(m)$              & $\tilde{O}(n/\sqrt{m})$   & $O(m^{3/2})$\istrut[1.5]{4.5}\\\hline
Fineman~\cite{Fineman18}     &   $\tilde{O}(n)$      & $\tilde{O}(n^{2/3})$              & $\tilde{O}(m)$\istrut[1.5]{4.5}\\\hline\hline
Hesse~\cite{Hesse03}       &   $O(mn^{1/17})$      & $\Omega(n^{1/17})$        & ---\istrut[1.5]{4.5}\\\hline
\rb{-3}{\bf new}     &   $O(n)$              & $\Omega(n^{1/6})$         & ---\istrut[1.5]{4.5}\\\cline{2-4}
                    &   $O(m)$              & $\Omega(n^{1/11})$        & ---\istrut[1.5]{4.5}\\\hline\hline
\end{tabular}
\caption{\label{tab:shortcut}Upper and Lower bounds on shortcutting sets.  The lower bounds are existential, and independent of computation time.}
\end{table}

\paragraph*{Additive Spanners.} Additive spanners with constant stretches
were discovered by Aingworth, Checkuri, Indyk, and Motwani~\cite{AingworthCIM99} (see also~\cite{DorHZ00,ElkinP01,BKMP10,Knudsen13}),
Chechik~\cite{Chechik13},
and Baswana, Kavitha, Mehlhorn, and Pettie~\cite{BKMP10} (see also~\cite{woodruff_lower_2006,Knudsen13}). The sparsest of these~\cite{BKMP10}
has size $O(n^{4/3})$ and stretch $+6$.
Abboud and Bodwin~\cite{abboud_bodwin_2017} 
showed that the $4/3$ exponent could not be improved, in the sense that
any $+n^{o(1)}$ spanner has size $\Omega(n^{4/3-o(1)})$, and that
any $\Omega(n^{4/3-\epsilon})$-size spanner has additive stretch $+\Omega(n^\delta)$, $\delta=\delta(\epsilon)$.
On the upper bound side, Pettie~\cite{Pettie09} showed that $O(n)$-size spanners could have additive stretch $+\tilde{O}(n^{9/16})$, and
Bodwin and Williams~\cite{BodwinW16} improved this to $O(\sqrt{n})$ for $O(n)$-size spanners and $O(n^{3/7})$ for $O(n^{1+o(1)})$-size spanners.
Abboud and Bodwin~\cite{abboud_bodwin_2017} extended their lower bound
to $O(n)$-size spanners, showing that they require 
stretch $+\Omega(n^{1/22})$.   Using our lower bound for shortcuttings as a starting place, we improve \cite{abboud_bodwin_2017} by giving an $+\Omega(n^{1/11})$ stretch lower bound for $O(n)$-size spanners.
See Table~\ref{tab:spanner}.

\begin{table}
\centering
\begin{tabular}{l@{\hcm[.3]}|l@{\hcm[.3]}|l@{\hcm[.3]}|@{\hcm[.3]}l}
\multicolumn{1}{l}{\bf Citation} &
\multicolumn{1}{@{\hcm[0]}l}{\bf Spanner Size} &
\multicolumn{1}{@{\hcm[0]}l}{\bf Additive Stretch} &
\multicolumn{1}{@{\hcm[-.1]}l}{\bf Remarks}  \\\hline\hline
Aingworth, Chekuri,\istrut[0]{4.5}         &   \rb{-2.5}{$O(n^{3/2})$}    &   \rb{-2.5}{$2$}     & \rb{-2.5}{See also~\cite{DorHZ00,ElkinP01,BKMP10,Knudsen13}}\\
Indyk, and Mowani~\cite{AingworthCIM99}\istrut[1.5]{4}   &                   &           &\\\hline
Chechik~\cite{Chechik13}             &   $\tilde{O}(n^{7/5})$    & $4$   & \istrut[1.5]{4.5}\\\hline
Baswana, Kavitha,\istrut[0]{4.5}           &   \rb{-2.5}{$O(n^{4/3})$}    &   \rb{-2.5}{$6$}     & \rb{-2.5}{See also~\cite{woodruff_lower_2006,Knudsen13}}\\
Mehlhorn, and Pettie~\cite{BKMP10}\istrut[1.5]{4}    &           &           & \\\hline
Pettie~\cite{Pettie09}  & $O(n^{1+\epsilon})$   &   $O(n^{9/16 - 7\epsilon/8})$\ \ \ \  & $0 \le \epsilon$\istrut[1.5]{4.5}\\\hline
Chechik~\cite{Chechik13} & $O(n^{20/17+\epsilon})$ & $O(n^{4/17 - 3\epsilon/2})$ & $0 \le \epsilon$\istrut[1.5]{4.5}\\\hline
\rb{-2.5}{Bodwin and Williams~\cite{BodwinW15}} & \rb{-2.5}{$O(n^{1+\epsilon})$}    &   $O(n^{1/2 - \epsilon/2})$\istrut[0]{4.5} & \rb{-2.5}{$0 \le \epsilon$}\\
                            &                       & $O(n^{2/3-5\epsilon/3})$\istrut[1.5]{4}  &\\\hline
                             &                        &   $O(n^{3/7-\epsilon})$   & $0 \le \epsilon \le 6/49$  \istrut[1.5]{4.5}\\\cline{3-4}
Bodwin and Williams~\cite{BodwinW16}  &  $O(n^{1+o(1)+\epsilon})$    &   $O(n^{3/5-12\epsilon/5})$   &  $6/49 \le \epsilon \le 2/13$ \istrut[1.5]{4.5}\\\cline{3-4}
                            &                                & $O(n^{3/7 - 9\epsilon/7})$ & $2/13 \le \epsilon < 1/3$\istrut[1.5]{4.5}\\\hline\hline
\rb{-3}{Abboud and Bodwin~\cite{abboud_bodwin_2017}}  &   $O(n^{4/3-\epsilon})$   &   $\Omega(n^\delta)$          &   $\delta=\delta(\epsilon)$\istrut[1.5]{4.5}\\\cline{2-4}
                            &   $O(n)$                  &   $\Omega(n^{1/22})$  &  \istrut[1.5]{4.5} \\\hline
{\bf new}                   &   $O(n)$                  &   $\Omega(n^{1/11})$  &\istrut[1.5]{4.5}\\\hline\hline
\end{tabular}
\caption{\label{tab:spanner}Upper and lower bounds on additive spanners.}
\end{table}

\begin{table}
\centering
\begin{tabular}{l@{\hcm[.3]}|l@{\hcm[.3]}|@{\hcm[.3]}l@{\hcm[.3]}|@{\hcm[.3]}l}
\multicolumn{1}{l}{\bf Citation} &
\multicolumn{1}{@{\hcm[0]}l}{\bf Emulator Size} &
\multicolumn{1}{@{\hcm[-.1]}l}{\bf Additive Stretch} &
\multicolumn{1}{@{\hcm[-.1]}l}{\bf Remarks}  \\\hline\hline
Aingworth, Chekuri,\istrut[0]{4.5}         &   \rb{-2.5}{$O(n^{3/2})$}    &   \rb{-2.5}{$2$}     & \rb{-2.5}{See also~\cite{DorHZ00,ElkinP01,BKMP10,Knudsen13}}\\
Indyk, and Mowani~\cite{AingworthCIM99}\istrut[1.5]{4}   &                   &           &\\\hline
Dor, Halperin, and Zwick~\cite{DorHZ00} &   $O(n^{4/3})$    & $4$   & \istrut[1.5]{4.5}\\\hline
Baswana, Kavitha,\istrut[0]{4.5}  & \rb{-2.5}{$O(n^{1+\epsilon})$}   & \rb{-2.5}{$O(n^{1/2-3\epsilon/2})$} & \rb{-2.5}{(not claimed in~\cite{BKMP10})}\\
Mehlhorn and Pettie~\cite{BKMP10}\istrut[1.5]{4} &                  &                           &\\\hline
Bodwin and Williams~\cite{BodwinW15}     & $O(n^{1+\epsilon})$               & $O(n^{1/3-2\epsilon/3})$ &\istrut[0]{4.5}\\\hline
Bodwin and Williams~\cite{BodwinW16}     & $O(n^{1+o(1)+\epsilon})$               & $O(n^{3/11 - 9\epsilon/11})$ & (conseq. of \cite[Thm.~5]{BodwinW16}) \istrut[0]{4.5}\\\hline
Pettie~\cite{Pettie09}                  & $O(n^{1+\epsilon})$               & $\tilde{O}(n^{1/4-3\epsilon/4})$ & (not claimed in~\cite{Pettie09})\istrut[0]{4.5}\\\hline\hline
Abboud and Bodwin~\cite{abboud_bodwin_2017}       & $O(n)$                            & $\Omega(n^{1/22})$    &\istrut[0]{4.5}\\\hline
{\bf new}                       & $O(n)$                            & $\Omega(n^{1/18})$\istrut[0]{4.5}\\\hline\hline
\end{tabular}
\caption{\label{tab:emulator}Upper and lower bounds on additive emulators.  Emulators with sublinear additive stretch~\cite{ThorupZ06,HuangP17,AbboudBP17} are not shown.}
\end{table}

\paragraph*{Additive Emulators.} Dor, Halperin, and Zwick~\cite{DorHZ00}
were the first to explicitly define the notion of an \emph{emulator}, 
and gave a $+4$ emulator with size $O(n^{4/3})$.  Abboud and Bodwin's~\cite{abboud_bodwin_2017} lower bound applies to emulators, i.e., we cannot go below the $4/3$ threshold without incurring polynomial 
additive stretch.  Bodwin and Williams~\cite{BodwinW15,BodwinW16}
pointed out that some spanner construtions~\cite{BKMP10} imply emulator bounds, and gave new constructions of emulators with size
$O(n)$ and stretch $+O(n^{1/3})$, 
and with size $O(n^{1+o(1)})$ and stretch $+O(n^{3/11})$.\footnote{This last result is a consequence of \cite[Thm.~5]{BodwinW16} and the fact 
that any pair set $P\subset V^2$ has a pair-wise emulator with size $|P|$.}  Here we observe that Pettie's~\cite{Pettie09} $+\tilde{O}(n^{9/16})$ spanner, when turned into an 
$O(n)$-size emulator, has stretch $+\tilde{O}(n^{1/4})$,
which is slightly better than the linear size emulators found in~\cite{BKMP10,BodwinW15,BodwinW16}.  We improve Abboud and Bodwin's~\cite{abboud_bodwin_2017} lower bound and show that any $O(n)$-size emulator has additive stretch $+\Omega(n^{1/18})$.  See Table~\ref{tab:emulator}.

Our emulator lower bounds are polynomially 
weaker than the spanner lower bounds. 
Although 
neither bound is likely sharp, this difference 
reflects the rule that emulators are probably more powerful than spanners.  For example, 
at sparsity $O(n^{4/3})$, the best known emulators~\cite{DorHZ00} are slightly 
better than spanners~\cite{BKMP10}.
Below the 4/3 threshold the best \emph{sublinear additive}
emulators~\cite{ThorupZ06,HuangP17} have size $O(n^{1+\frac{1}{2^{k+1}-1}})$ and stretch function $d+O(d^{1-1/k})$.\footnote{I.e., vertices initially at 
distance $d$ are stretched to $d+O(d^{1-1/k})$.}
Abboud, Bodwin, and Pettie~\cite{AbboudBP17} showed that this
tradeoff is optimal for emulators, but the best known 
sublinear additive spanners~\cite{Pettie09,Chechik13} are polynomially worse.

\newcommand{\edited}[1]{{\color{blue}{#1}}}

There are a certain range of parameters where emulators are known to be polynomially sparser than spanners. 
For pairwise distance preservers, Bodwin~\cite{Bodwin17} showed that whenever $\omega(n^{1/2})=|P|=o(n^{2-o(1)})$, any pairwise distance preserver has an $\omega(n+|P|)$ lower bound, which is worse than the trivial 
distance preserving emulator with size $|P|$.  
A similar separation holds for 
source-wise distance preservers, where the goal is to exactly 
preserve distances between all vertex pairs in $S\subset V$.
A trivial source-wise
\emph{emulator} has size $|S|^2$, e.g., $O(n)$ for $|S|=\sqrt{n}$,
but source-wise \emph{spanners} with size $O(n)$ only exist 
for $|S|=O(n^{1/4})$~\cite{CoppersmithE06,Bodwin17}.

\paragraph*{Organization.} 
In Section~\ref{sect:shortcut} we present diameter lower bounds
for shortcut sets of size $O(n)$ and $O(m)$.  Section~\ref{sect:spanner-emulator} modifies the construction
to give lower bounds on additive spanners and additive emulators.
We conclude with some remarks in Section~\ref{sect:conclusion}.

\section{Lower Bounds on Shortcutting Digraphs}\label{sect:shortcut}

\subsection{Using \texorpdfstring{$O(n)$}{O(n)} Shortcuts}\label{subsection:using-n-shortcuts}

Existentially, the best known upper bound on $O(n)$-size shortcut sets 
is the trivial $\tilde{O}(\sqrt{n})$ bound.  Theorem~\ref{theorem:n-shortcuts}
shows that we cannot go below $\Omega(n^{1/6})$.

\begin{theorem}\label{theorem:n-shortcuts}
There exists a directed graph $G$ with $n$ vertices, such that for any shortcut set $E'$ with size $O(n)$, 
the graph $(V, E\cup E')$ has diameter $\Omega(n^{1/6})$.
\end{theorem}

The remainder of Section~\ref{subsection:using-n-shortcuts} constitutes a proof of Theorem~\ref{theorem:n-shortcuts}.  We begin by defining the vertex set and edge set of $G$, and its \emph{critical pairs}.

\paragraph*{Vertices.}
The vertex set of $G$ is partitioned into $D+1$ \emph{layers} 
numbered $0$ through $D$.  
Define $B_d(\rho)$ to be the set of all lattice points
in $\mathbb{Z}^d$ within Euclidean distance $\rho$ of the origin.
In the calculations below we treat $d$ as a constant.
For each $k\in\{0,\ldots,D\}$, 
layer-$k$ vertices are identified with lattice points in
$B_d(R + kr)$, where $r,R$ are parameters of the construction.
A vertex can be represented by a pair $(a,k)$, where $a\in B_{d}(R+rk)$.
We want the size of all layers to be the same, up to a constant factor.  To that end we fix $R=drD$, so the total number of vertices is \begin{align*}
    n &\approx \eta_d R^d\left( 1^d + \left(1+\frac{r}{R}\right)^d +
    \cdots + \left(1+\frac{rD}{R}\right)^d\right)\\
    &=  \eta_d R^d\left(1^d + \left(1+\frac{1}{dD}\right)^d +
    \cdots + \left(1+\frac{1}{d}\right)^d\right) \;=\; \Theta \left(R^d D\right) & \mbox{(By definition of $R$)}
\end{align*}
where $\eta_d=\frac{1}{\sqrt{2\pi d}}\left(\frac{2\pi e}{d}\right)^{d/2}$ is the ratio of volume between a $d$-dimentional ball and its enclosing $d$-dimentional cube.

\newcommand{\ConvexHull}{\mathcal{V}}

\paragraph*{Edges.} Define $\ConvexHull_d(r)$ to be the set of all lattice
points at the corners of the convex hull of $B_d(r)$.
(This excludes points that happen to lie on the boundary, but in the interior
of one of its faces.)
We treat elements of $\ConvexHull_d(r)$ as vectors.
For each layer-$k$ vertex $(a,k)$, $k\in\{0,\ldots,D-1\}$, and each
vector $v\in\ConvexHull_d(r)$, we include a directed edge
$((a,k), (a+v,k+1))$.  All edges in $G$ are of this form.

\paragraph*{Critical Pairs.} The critical pair set is defined to be
\[
P = \{((a,0), (a+Dv,D)) \;|\; a\in B_d(R) \mbox{ and } v\in \ConvexHull_d(r)\}
\]
Each such pair has a corresponding path of length $D$, namely
$(a,0)\rightarrow (a+v,1) \rightarrow \cdots \rightarrow (a+Dv,D)$.
Lemma~\ref{lemma:critical-pairs} shows that this path is unique.
It was first proved by Hesse~\cite{Hesse03} and independently by
Coppersmith and Elkin~\cite{CoppersmithE06}.  (Both proofs are inspired by 
Behrend's~\cite{Behrend46} construction of arithmetic progression-free sets, which
uses $\ell_2$ balls rather than convex hulls.)

\begin{lemma}\label{lemma:critical-pairs} 
(cf.~\cite{Hesse03,CoppersmithE06})
The set of critical pairs $P$ have the following properties:
\begin{itemize}
    \item For all $(x, y)\in P$, there is a unique path from $x$ to $y$ in $G$.
    \item For any two distinct pairs $(x_1, y_1)$ and $(x_2, y_2)\in P$,
    their unique paths share no edge and at most one vertex.
    \item $|P|=\Theta(R^d r^{d\frac{d-1}{d+1}})$.
\end{itemize}
\end{lemma}

\begin{proof}
For the first claim, 
let $x = (a, 0)$ and $v\in\ConvexHull_d(r)$ be the vector for which
    $y=(a+Dv, D)$. One path from $x$ to $y$ exists by construction.
    Let $\ConvexHull_d(r) = \{v_1, v_2, \ldots, v_s\}$. 
    Suppose there exists another path from $x$ to $y$.
    It must have length $D$ because all edges join consecutive layers.
    Every edge on this path corresponds to a vector $v_i$, which implies that 
    $Dv$ can be represented as a linear combination $k_1v_1+k_2v_2+\cdots + k_sv_s$, where $k_1+\cdots + k_s = D$ and $k_i\ge 0$. This implies that $v$ is a non-trivial convex combination of the vectors in $\ConvexHull_d(r)$, which contradicts 
    the fact that $\ConvexHull_d(r)$ is a strictly convex set.
    
The second claim follows from the fact that any \emph{edge} in the unique $x_1$-to-$y_1$ path uniquely identifies both $x_1$ and $y_1$.
    
For the last claim, we can express the number of critical pairs as $|P|=|B_d(R)|\cdot |\ConvexHull_d(r)|$. From B\'ar\'any and Larman~\cite{Barany98}, for any constant dimension $d$, we have $|\ConvexHull_d(r)|=\Theta(r^{d\frac{d-1}{d+1}})$.
\end{proof}

\begin{lemma}\label{lemma:shortcut-pair-relation}
Let $E'$ be a shortcut set for $G=(V, E)$. If the diameter of $G' = (V, E\cup E')$ 
is strictly less than $D$, then $|E'|\ge |P|$.
\end{lemma}

\begin{proof}
Every path in $G'$ corresponds to some path in $G$.  However, for pairs
in $P$, there is only one path in $G$, hence any shortcut in $E'$ useful
for a pair $(x,y)\in P$ must have both endpoints on the unique $x$-$y$ path in $G$.
By Lemma~\ref{lemma:critical-pairs}, two such paths for pairs in $P$ share
no common edges, hence each shortcut can only be useful for at most 
\underline{one} pair in $P$.  If $|E'| < |P|$ then some pair $(x,y)\in P$ must still be at distance $D$ in $G'$.
\end{proof}

\begin{proof}[Proof of Theorem~\ref{theorem:n-shortcuts}.]
By Lemma~\ref{lemma:shortcut-pair-relation}, if $|P| =\Omega(n)$, then any shortcut set that makes the diameter $<D$ has size $\Omega(n)$. In order to have $|P| = \Omega(n)$, it suffices to let $r^{d\frac{d-1}{d+1}}\ge D$. This implies $r\ge D^{\frac{d+1}{d(d-1)}}$.
From the construction, by fixing $d$ as a constant, we have
\[
n \;=\; \Theta(R^d D)
  \;=\; \Theta((rD)^d D)
  \;=\; \Omega(D^{1+d+\frac{d+1}{d-1}}). 
\]
Therefore, the diameter is $D = O\left(n^{1/\left(1+d+\frac{d+1}{d-1}\right)}\right)$. 
We can maximize $D=\Theta(n^{1/6})$ in one of two ways, by setting $d=2$,
$r = \Theta(n^{1/4})$, and $R=\Theta(n^{5/12})$, 
or $d=3$, $r=\Theta(n^{1/9})$, and $R=\Theta(n^{5/18})$.
In either case, the construction leads to a graph with very similar structure: 
the number of vertices in each layer is $\Theta(n^{5/6})$, and the out degrees
of each vertex are $\Theta(n^{1/6})$.
\end{proof}

Theorem~\ref{theorem:n-shortcuts} is indifferent between $d=2$ and $d=3$
but that is only because the size of the shortcut set is precisely $O(n)$.
When we allow it to be $O(n^{1+\epsilon})$, for $\epsilon>0$, 
there is generally one optimum dimension.
\begin{corollary}
Fix an $\epsilon \in [0,1)$ and let $d$ be an integer such that 
$\epsilon \in [0, \frac{d-1}{d+1}]$. 
There exists a directed graph $G$ with $n$ vertices, such that for any 
shortcut set $E'$ with $O(n^{1+\epsilon})$ shortcuts, the graph $(V, E\cup E')$ has diameter $\Omega(n^{\left(1-\frac{d+1}{d-1}\epsilon\right)/\left(1+d+\frac{d+1}{d-1}\right)})$. In particular, by setting $d=3$ the diameter lower bound becomes $\Omega(n^{\frac16-\frac13\epsilon})$.
\end{corollary}

\begin{proof}
In order to have $|P| > n^{1+\epsilon}$, it suffices to let $r^{d\frac{d-1}{d+1}} \ge Dn^\epsilon$. Hence, we have
\begin{align*}
    n^{1-{\frac{d+1}{d-1}\epsilon}} &= \Theta(R^dDn^{-\frac{d+1}{d-1}\epsilon})\\
     &= \Omega(r^d D^{1+d} n^{-\frac{d+1}{d-1}\epsilon})    & (R=\Theta(rD))\\
     &= \Omega(D^{1+d+\frac{d+1}{d-1}})             & (r^d \ge (Dn^\epsilon)^{\frac{d+1}{d-1}})
\end{align*}
\end{proof}

\subsection{Using \texorpdfstring{$O(m)$}{O(m)} Shortcuts}\label{subsection:using-m-shortcuts}

Let $G_{(d, r, D)}$ denote the layered graph constructed in Section~\ref{subsection:using-n-shortcuts} with parameters $d, D, r,$ and $R=drD$,
and let $P_G$ be its critical pair set. The total number of edges $m=\Theta(n|\ConvexHull_d(r)|)$ is always larger than $|P_G|=\Theta(\frac{n}{D} |\ConvexHull_d(r)|)$ by a factor of $D$. In order to get a lower bound for $O(m)$ shortcuts, we use a Cartesian product combining two such graphs layer by layer, forming a sparser graph.  This transformation
was discovered by Hesse~\cite{Hesse03} and rediscovered by Abboud and Bodwin~\cite{abboud_bodwin_2017}.

Let $G_1=G_{(d_1, r_1, D)}$ and $G_2=G_{(d_2, r_2, D)}$ be two graphs with the same number of layers, namely $D+1$. 
The product graph $G_1\otimes G_2$ is defined below.

\paragraph*{Vertices.}

The product graph has $2D+1$ vertex layers numbered $0, \ldots, 2D$. 
The vertex set of layer $i$ is $\{(x, y, i)\ |\ x\in B_{d_1}(R_1+\left\lceil\frac{i}{2}\right\rceil r_1), y\in B_{d_2}(R_2 + \left\lfloor\frac{i}{2}\right\rfloor r_2)\}$.
Since we set $R_j=d_jr_jD$, the total number of vertices is $\Theta\left(R_1^{d_1}R_2^{d_2}D\right)$.

\paragraph*{Edges.}
Let $(x,y,i)$ be a vertex in layer $i$.  If $i$ is even,
then for every vector $v\in \ConvexHull_{d_1}(r_1)$ we include
an edge $((x,y,i), (x+v,y,i+1))$.  If $i$ is odd,
then for every vector $w\in \ConvexHull_{d_2}(r_2)$, we include an edge $((x,y,i), (x,y+w,i+1))$.
The total number of edges in the product 
graph is then $\Theta\left(
R_1^{d_1}R_2^{d_2}D\left(r_1^{d_1\frac{d_1-1}{d_1+1}}+r_2^{d_2\frac{d_2-1}{d_2+1}}\right)
\right)$.

\paragraph*{Critical Pairs.}

By combining two graphs, we are able to construct a larger set of critical pairs, as follows.
\[
P = \{((a,b,0), (a+Dv, b+Dw, 2D)) \;|\; a\in B_{d_1}(R_1), b\in B_{d_2}(R_2), v\in \ConvexHull_{d_1}(r_1), w\in \ConvexHull_{d_2}(r_2)\}
\]
In other words, a pair in $P$ can be viewed as the product of two pairs
$((a,0),(a+Dv,D)) \in P_{G_1}$ and $((b,0),(b+Dw,D)) \in P_{G_2}$.

\begin{lemma}\label{lemma:uniquepath}
For any $a\in B_{d_1}(R_1)$, $b\in B_{d_2}(R_2)$, $v\in \ConvexHull_{d_1}(r_1)$, and $w\in \ConvexHull_{d_2}(r_2)$, there is a unique path from $(a, b, 0)$ to 
$(a+Dv, b+Dw, 2D)$.
\end{lemma}

\begin{proof}
Every path in $G_1\otimes G_2$ from layer $0$ to layer $2D$ corresponds to two paths from layers 0 to $D$ in 
$G_1$ and $G_2$, respectively.
It follows from Lemma~\ref{lemma:critical-pairs} that
\[
(a,b,0)\rightarrow (a+v,b,1) \rightarrow (a+v,b+w,2) \rightarrow \cdots \rightarrow (a+Dv, b+Dw, 2D)
\]
is a unique path in $G_1\otimes G_2$.
\end{proof}

In $G_1\otimes G_2$ it is no longer true that 
pairs in $P$ have edge-disjoint paths. 
They may intersect at just one edge.

\begin{lemma}\label{lemma:intersection-two-paths-ii}
Consider two pairs $(x_1, y_1)$ and $(x_2, y_2)\in P$. Let $P_1$ and $P_2$ be the unique shortest paths in the combined graph from $x_1$ to $y_1$ and from $x_2$ to $y_2$. Then, $P_1\cap P_2$ contains at most one edge.
\end{lemma}

\begin{proof}
Any two \emph{non-adjacent} vertices on the unique $x_1$-$y_1$ path
uniquely identify $x_1$ and $y_1$.  Thus, two such paths can intersect in at most 2 (consecutive) vertices, and hence one edge.
\end{proof}

\begin{lemma}\label{lemma:shortcut-pair-relation-ii}
Let $E'$ be a shortcut set on $G=(V, E)$. If the diameter of $(V, E\cup E')$ 
is strictly less than $2D$, then $|E'|\ge |P|$.
\end{lemma}

\begin{proof}
Assume the diameter of $(V, E\cup E')$ is strictly less than $2D$. 
Every useful shortcut connects vertices that are at distance at least 2.
By Lemma~\ref{lemma:intersection-two-paths-ii}, such a shortcut can only be useful for
one pair in $P$. Thus, if the diameter of $(V,E\cup E')$ is less than $2D$, $|E'|\ge |P|$.
\end{proof}

By construction, the size of $|P|$ is  
\[
|P|=\Theta\left(R_1^{d_1}R_2^{d_2}|\ConvexHull_{d_1}(r_1)||\ConvexHull_{d_2}(r_2)|\right) = \Theta\left(R_1^{d_1}R_2^{d_2}r_1^{d_1\frac{d_1-1}{d_1+1}}r_2^{d_2\frac{d_2-1}{d_2+1}}\right). 
\]

\begin{theorem}\label{theorem:m-shortcuts}
There exists a directed graph $G$ with $n$ vertices and $m$ edges
such that for any shortcut set $E'$ with size $O(m)$, 
the graph $(V, E\cup E')$ has diameter $\Omega(n^{1/11})$.
\end{theorem}

\begin{proof}
If we set $|P|=\Omega(m)$, by Lemma~\ref{lemma:shortcut-pair-relation-ii}, any shortcut set $E'$ with $O(m)$ shortcuts has diameter $\Omega(D)$. In order to ensure $|P|=\Omega(m)$, it suffices to set $r_1^{d_1\frac{d_1-1}{d_1+1}} \ge r_2^{d_2\frac{d_2-1}{d_2+1}} \ge D$. Hence,
\begin{align*}
n &= \Theta(R_1^{d_1}R_2^{d_2}D) \\
&= \Theta\left(r_1^{d_1}r_2^{d_2}D^{d_1+d_2+1}\right) \tag{$R_j=d_jr_jD$}\\
&= \Omega\left(D^{\frac{d_1+1}{d_1-1}} D^{\frac{d_2+1}{d_2-1}} D^{d_1+d_2+1}\right) \tag{plugging in relation between $r_j$ and $d_j,D$}\\
&= \Omega\left(D^{\frac{d_1+1}{d_1-1}+\frac{d_2+1}{d_2-1} + d_1+d_2+1}\right)
\end{align*}
The exponent is minimized when $d_1$ and $d_2$ are either $2$ or $3$,
so we get $n=\Omega(D^{11})$ and hence $D=O(n^{1/11})$. 
In particular, by setting $d_1=d_2=2$
we have $D=\Theta(n^{1/11})$, 
$r_1=r_2=\Theta(n^{3/22})$ and 
$R_1=R_2=\Theta(n^{5/22})$, and by setting $d_1=d_2=3$
we have $D=\Theta(n^{1/11})$,
$r_1=r_2=\Theta(n^{2/33})$, and
$R_1=R_2=\Theta(n^{5/33})$.
\end{proof}

\section{Lower Bounds on Additive Spanners and Emulators}\label{sect:spanner-emulator}

We now establish better bounds on $O(n)$-size additive spanners
and emulators.  In Section~\ref{subsection:n-spanners}, we 
give an $+\Omega(n^{1/13})$ stretch lower bound on spanners. 
Using a different construction, we improve this in Section~\ref{subsection:improved-n-spanner} to $+\Omega(n^{1/11})$.
In Section~\ref{subsection:n-emulators} we show how to adapt
the $+\Omega(n^{1/13})$-spanner lower bound from Section~\ref{subsection:n-spanners} to prove that
$O(n)$-size emulators have stretch $+\Omega(n^{1/18})$.

Recall the definition of additive spanners and emulators.

\begin{definition}
Let $G=(V, E)$ be an unweighted undirected graph. 
A \emph{subgraph} $H=(V, E')$, $E'\subseteq E$, 
is said to be a \emph{spanner for $G$ with additive stretch $\beta$} 
if for any two vertices $u, v\in V$, 
\[
\mathrm{dist}_H(u, v) \le \mathrm{dist}_G(u, v) + \beta.
\]
A weighted graph $H=(V,E',w)$ is an \emph{emulator for $G$ with additive stretch $\beta$} if
\[
\mathrm{dist}_G(u, v) \le \mathrm{dist}_H(u, v) \le \mathrm{dist}_G(u, v) + \beta.
\]
Observe that we can assume w.l.o.g.~that if $(u,v)\in E'$ then
$w(u,v) = \mathrm{dist}_G(u,v)$.
\end{definition}

\subsection{\texorpdfstring{$O(n)$}{O(n)}-Size Spanners}\label{subsection:n-spanners}

By combining the technique of Abboud and Bodwin~\cite{abboud_bodwin_2017}
with the graphs constructed in Section~\ref{subsection:using-m-shortcuts},
we improve the $+\Omega(n^{1/22})$ lower bound of~\cite{abboud_bodwin_2017}
to $+\Omega(n^{1/13})$ for $O(n)$-size spanners.

\begin{theorem}\label{theorem:n-spanner}
There exists an undirected graph $G$ with $n$ vertices, 
such that any spanner for $G$ with $O(n)$ edges has $+\Omega(n^{1/13})$ additive stretch.
\end{theorem}

In this section we regard $G_{(d, r, D)}$ to be an \emph{undirected} graph.
We begin with the undirected graph 
$G_0 = G_{(d_1, r_1, D)}\otimes G_{(d_2, r_2, D)}$,
then modify it in the \emph{edge subdivision step} 
and the \emph{clique replacement step} to obtain $G$.

\paragraph*{The Edge Subdivision Step.}

Every edge in $G_0$ is subdivided into $D$ edges, yielding $G_E$. 
This step makes the graph very sparse 
since most of the vertices in $G_E$ now
have degree 2.

\paragraph*{The Clique Replacement Step.}

Consider a vertex $u$ in $G_E$ that comes from one of the 
interior layers of $G_0$, i.e., layers $1,\ldots,2D-1$, not 0 or $2D$.
Note that $u$ has degree $\delta_1+\delta_2$,
with $\delta_1=\Theta\left(r_1^{d_1\frac{d_1-1}{d_1+1}}\right)$
edges leading to the preceding layer
and $\delta_2=\Theta\left(r_2^{d_2\frac{d_2-1}{d_2+1}}\right)$
edges leading to the following layer (or vice versa). 
We replace each such $u$ with a complete bipartite clique $K_{\delta_1, \delta_2}$,
where each clique vertex becomes attached to one non-clique edge formerly attached to $u$.
The final graph is denoted $G$.

\paragraph*{Critical Pairs.}

The set $P$ of \emph{critical pairs} for $G$ is 
identical to the set of critical pairs for $G_0$.
For each $(x,y)\in P$, the unique $x$-$y$ path in $G$ is called a \emph{critical path}.
From the construction, the number of vertices, edges,
and critical pairs in $G$ is
\begin{align}
    n &=\Theta\left( R_1^{d_1}R_2^{d_2} D^2(\delta_1+\delta_2) \right).\label{eq:spanner-construction-n}\\
    m &=\Theta\left( R_1^{d_1} R_2^{d_2} D (D\delta_1+D\delta_2+\delta_1\delta_2) \right).\label{eq:spanner-construction-m}\\
    |P| &=\Theta\left( R_1^{d_1}R_2^{d_2} \delta_1\delta_2 \right).\label{eq:spanner-construction-p}
\end{align}

Lemma~\ref{lemma:clique-edge-path-relation} is 
key to relating the size of the spanner with the pair set $P$.

\begin{lemma}\label{lemma:clique-edge-path-relation}
Every clique edge belongs to at most one critical path.
\end{lemma}

\begin{proof}
Every clique has $\delta_1$ vertices on one side and $\delta_2$ vertices on the other side.
Each vertex on the $\delta_1$ side corresponds to a vector $v\in\ConvexHull_{d_1}(r_1)$
and each vertex on the $\delta_2$ side corresponds to a vector $w\in\ConvexHull_{d_2}(r_2)$. 
Each clique edge uniquely determines a pair of vectors $(v, w)$, and 
hence exactly one critical pair in $P$.
\end{proof}

\begin{lemma}\label{lemma:spanner-construction-lowerbound}
Every spanner of $G$ with additive stretch $+(2D-1)$
must contain at least $D|P|$ clique edges.
\end{lemma}

\begin{proof}
For the sake of contradiction suppose there exists a spanner $H$ containing at most $D|P|-1$ clique edges.
By the pigeonhole principle there exists a pair $(x, y)\in P$ such that at least 
$D$ clique edges are \emph{missing} in $H$.

Let $P_{(x, y)}$ be the unique shortest path from $x$ to $y$ in $G$, 
and let $P'_{(x, y)}$ be a shortest path from $x$ to $y$ in $H$. 
Since $G_0$ is formed from $G$ by contracting all bipartite cliques 
and replacing subdivided edges with single edges, we can 
apply the same operations on $P'_{(x, y)}$ to get a path 
$P''_{(x, y)}$ in $G_0$. We now consider two cases:
\begin{itemize}
    \item If $P''_{(x, y)}$ is the unique shortest path from $x$ to $y$ in $G_0$, then
    $P'_{(x, y)}$ suffers at least a $+2$ stretch on each of the $D$ missing clique edges, 
    so $|P'_{(x, y)}| \ge |P_{(x, y)}| + 2D$.
    \item If $P''_{(x, y)}$ is not the unique shortest path from $x$ to $y$ in $G_0$, then it
    must traverse at least two more edges than the shortest $x$-$y$ path in $G_0$ (because $G_0$ is bipartite), each of which is subdivided $D$ times in the formation of $G$.  Thus
    $|P'_{(x, y)}| \ge |P_{(x, y)}| + 2D$.
\end{itemize}
In either case, $P'_{(x, y)}$ has at least $+2D$ additive stretch and 
$H$ cannot be a $+(2D-1)$ spanner.
\end{proof}

\begin{proof}[Proof of Theorem~\ref{theorem:n-spanner}]
The goal is to have parameters set up so that $D|P| = \Omega(n)$,
so that we can apply Lemma~\ref{lemma:spanner-construction-lowerbound}.
Without loss of generality $\delta_1 \ge \delta_2$. 
By comparing \eqref{eq:spanner-construction-n} with \eqref{eq:spanner-construction-p}, 
it suffices to set $\delta_1\ge \delta_2 \ge D$.  We can express the number of vertices
in terms of $D$ as follows:
\begin{align*}
    n &= \Theta\left(R_1^{d_1}R_2^{d_2}D^2\delta_1\right)\\
    &=\Omega\left((r_1D)^{d_1}(r_2D)^{d_2}D^3\right) \tag{$\delta_1\ge \delta_2\ge D$}\\
    &=\Omega\left(\left(\delta_1^{\frac{d_1+1}{d_1(d_1-1)}}D\right)^{d_1}
    \left(\delta_2^{\frac{d_2+1}{d_2(d_2-1)}}D\right)^{d_2} D^3\right) \tag{by definition of $\delta_1$ and $\delta_2$}\\
    &=\Omega\left(D^{\frac{d_1+1}{d_1-1}+d_1+\frac{d_2+1}{d_2-1}+d_2+3}\right) \tag{$\delta_1\ge \delta_2\ge D$}
\end{align*}

The exponent is minimized when $d_1$ and $d_2$ are either $2$ or $3$,
so $n=\Omega(D^{13})$ and hence the additive stretch is 
$D=O(n^{1/13})$. 
When $d_1=d_2=2$ we have 
$D=\Theta(n^{1/13})$, 
$r=\Theta(n^{3/26})$ and 
$R=\Theta(n^{5/26})$, and when $d_1=d_2=3$ we have
$D=\Theta(n^{1/13})$,
$r=\Theta(n^{2/39})$, and
$R=\Theta(n^{5/39})$.
\end{proof}

\begin{corollary}\label{cor:spanner-lb}
Fix an $\epsilon\in [0,1/3)$ and let $d$ be an integer such that $\epsilon \in \left[0, \frac{d-1}{3d+1}\right]$. There exists a graph $G$ with $n$ vertices such that any spanner $H\subseteq G$ 
with $O(n^{1+\epsilon})$ edges has additive stretch $+\Omega\left(n^{\left(1-\frac{3d+1}{d-1}\epsilon\right)/\left(3+2d+2\frac{d+1}{d-1}\right)}\right)$. In particular, by setting $d=3$ the additive stretch becomes $\Omega(n^{\frac{1}{13}-\frac{5}{13}\epsilon})$.
\end{corollary}

\subsection{An Improved \texorpdfstring{$O(n)$}{O(n)}-Size 
Spanner Lower Bound}\label{subsection:improved-n-spanner}

The construction from Section~\ref{subsection:n-spanners} is versatile,
inasmuch as it extends to polynomial densities (Corollary~\ref{cor:spanner-lb}) and emulator lower bounds (Section~\ref{subsection:n-emulators}).  
However, it
can be improved, slightly, for the specific case of $O(n)$-size 
additive spanners.  
It turns out that the the Cartesion product step (generating $G_0$ from $G_{(d_1, r_1, D)}\otimes G_{(d_2, r_2, D)}$)
is inefficient, 
and that we can do better with a simple replacement step.

By its nature, the proof of Theorem~\ref{theorem:n-spanner-harder} needs to
explictly keep track of the leading absolute constant 
in the size of the spanner, i.e., it has at most
$c_0n = O(n)$ edges.  (In contrast, the proof of Theorem~\ref{theorem:n-spanner} can easily accommodate
any $O(n)$-size bound by tweaking $r,R,D$ by constant factors.)

\begin{theorem}\label{theorem:n-spanner-harder}
For any parameter $c_0>1$ and sufficiently large $n$
there exists an undirected $n$-vertex graph $G$ 
such that any spanner for $G$ with at most $c_0n$ 
edges has $+\Omega(n^{1/11}c_0^{-18/11})$ additive stretch.
\end{theorem}

In Lemmas~\ref{lemma:inner-graph-construction} and \ref{lemma:spanner-outer-base-graph} we construct the \emph{inner}
and \emph{outer} graphs, then discuss how to 
combine them using a substitution product.

\begin{lemma}[Inner Graph Construction]\label{lemma:inner-graph-construction}
Fix a parameter $c>1$. 
There exists sufficiently large $q, L$ 
such that $q = \Theta(L^2c^6)$ and a graph 
$G_I = (V_I, E_I)$ with a set of critical pairs 
$P_I \subseteq V_I \times V_I$ satisfying the following.
\begin{enumerate}
    \item $|V_I|\le qL$.
    \item $|P_I|\ge q$.
    \item $\forall (u, v)\in P_I$, the shortest path between $u$ and $v$ is unique and has length $Lc$.
    \item $\forall (u_1, v_1), (u_2, v_2)\in P$, the unique shortest paths between $u_1$ and $v_1$ and between $u_2$ and $v_2$ are edge-disjoint. Moreover, $\dist_{G_I}(u_1, v_2) \ge Lc$ and $\dist_{G_I}(u_2, v_1)\ge Lc$. (As a consequence, $|E_I|\ge cqL$.)
\end{enumerate}
\end{lemma}

\begin{proof}
We use almost the same construction as in Theorem~\ref{theorem:n-shortcuts}, except that the graph will be undirected and we will pay closer attention to the density.
In this proof $d=O(1)$ represents
the density of the graph, not the geometric dimension.
We will ultimately choose $d=4c$. 
The graph $G_I$ we construct consists of $Lc+1$ layers, numbered by $0, 1, \ldots, Lc$.

Recall that $\ConvexHull_2(r)$ is the set of all lattice points at the corners of the convex hull of $B_2(r)$. 
Let $\eta_2=\Theta(1)$ be the ratio between the area of a circle and the area of its circumscribed square, and 
let $\xi_d = \Theta(1)$ be such that
$|\ConvexHull_2(\xi_d d^{3/2})| \ge d/(4\eta_2)$.
(It follows from~\cite{Barany98} that $\xi_d=\Theta(1)$.)
On the $k$-th layer, the vertices are labelled by $(a, k)$ where 
$a\in B_2(\sqrt{q/d} + k\xi_d d^{3/2})$.
For each layer-$k$ vertex $(a, k)$, $k\in\{0, \ldots, Lc-1\}$ and each vector $v\in \ConvexHull_2(\xi_d d^{3/2})$, we connect an (undirected) edge between $(a, k)$ and $(a+v, k+1)$.

By choosing $\sqrt{q/d}=(Lc)\xi_d d^{3/2}$, 
we have 
$q=L^2c^2\xi_d^2 d^4 = \Theta(L^2c^6)$
and the total number of vertices in $G$ can be 
upper bounded by the number of layers times the size of the last layer: $(Lc)(2\sqrt{q/d})^2 = 4qLc/d$. 
Thus, condition 1 is satisfied whenever $d\ge 4c$.

Define 
\[
P_I =\left\{ ((a, 0), (a+(Lc)v, (Lc))) \ \bigg|\ a\in B_2(\sqrt{q/d}) \text{ and } v\in \ConvexHull_2(\xi_d d^{3/2}) \right\}.
\]
We have that $|P_I| = |B_2(\sqrt{q/d})|\cdot |\ConvexHull_2(\xi_d d^{3/2})| \ge 4\eta_2(q/d)\cdot d/(4\eta_2) = q$, so condition 2 is satisfied.

Now, for each pair of vertices $((a, 0), (a+(Lc)v, (Lc)))$ in $P$, there is an unique shortest path from $(a, 0)$ to $(a+(Lc)v, (Lc))$ by Lemma~\ref{lemma:critical-pairs}. Moreover, since the graph is a layered graph, any path from a vertex in the $0$-th layer to \emph{any} vertex in the $(Lc)$-th layer has length at least $Lc$, satisfying constraints 3 and 4.
\end{proof}

Again, we use a similar construction to  Theorem~\ref{theorem:n-shortcuts} to obtain our outer graph.

\begin{lemma}[Outer Graph Construction, 3D version]~\label{lemma:spanner-outer-base-graph}
For any given $q, L\in \mathbb{N}$, there exists an undirected graph $G_0=(V_0,E_0)$ with a set of critical pairs $P\subseteq V_0\times V_0$ satisfying:
\begin{enumerate}
    \item $|V_0| = \Theta(L^4q^2)$.
    \item $|P| = \Theta(L^3q^3)$.
    \item $\forall (u, v)\in P$, the shortest path between $u$ and $v$ (denoted by $P_{uv}$) is unique. Moreover, $P_{uv}$ has length exactly $L$.
    \item $\forall (u_1, v_1), (u_2, v_2)\in P$, the unique shortest paths between $u_1$ and $v_1$ and between $u_2$ and $v_2$ are edge-disjoint.
\end{enumerate}
\end{lemma}

\begin{proof}
Consider the following $(L+1)$-layer graph. 
Vertices in the $k$-th layer are identified with points in 
the 3-dimensional integer lattice inside the 
ball of radius $Lr+kr$ around the origin. 
Here $r$ is the minimum value such that 
$|\ConvexHull_3(r)|\ge q$. From B\'ar\'any and Larman~\cite{Barany98} we have $r=\Theta(q^{2/3})$.

We label each vertex with its coordinate and its layer number: $(a, k)\in B_3(Lr+kr)\times [L+1]$.
Fix an arbitrary subset $\ConvexHull'_3(r)\subseteq \ConvexHull_3(r)$ of any $q$ vectors.
For each vertex $(a, k)$ in the $k$-th layer ($0\le k < L$), and for every vector $x\in \ConvexHull'_3(r)$, the edge $((a, k), (a+v, k+1))$ is added to the graph.

It is straightforward to check that 
$|V_0|\approx \sum_{k=0}^L \eta_3 (2(Lr+kr))^3 = \Theta(L^4q^2)$. For each vector $v\in\ConvexHull'_3(r)$ and each layer-0 vertex $(a, 0)$, the vertex pair $((a, 0), (a+Lv, L))$ is added to the critical pair set $P$, 
hence $|P|=\Theta((Lr)^3q) = \Theta(L^3q^3)$. 
By the same argument as in the proof of Lemma~\ref{lemma:critical-pairs}, there is exactly one shortest path of length $L$ connecting $(a, 0)$ and $(a+Lv, L)$. Moreover, no edge belongs to more than one critical path.
\end{proof}

Recall that we are aiming for lower bounds against spanners
with size $c_0n$.  Once $c_0$ is fixed, we choose a $c=\Theta(c_0)$
and invoke Lemma~\ref{lemma:inner-graph-construction} to construct
an inner graph $G_I$ with parameters $q,L$.  Once $q,L$ are fixed
we invoke Lemma~\ref{lemma:spanner-outer-base-graph} to build
the outer graph $G_0$.  Our final graph $G$ is formed from $G_0,G_I$
through the \emph{inner graph replacement step} and the \emph{edge subdivision step}, as follows.

\paragraph*{Inner Graph Replacement Step.}
For every vertex $(a, k)\in G_0$ $(0 < k < L)$, 
we replace $(a, k)$ with a copy of inner graph $G_I$ as follows.

Recall that the critical pair set for $G_I$ has size $q$.  We regard the sources of these $q$ pairs to be \emph{input ports} and the sinks to be \emph{output ports}.
Let $G_{I, (a,k)}$ be the copy of $G_I$ substituted for $(a,k)$ in the outer graph.  
For each $v_i\in \ConvexHull'_3(r)=\{v_1, v_2, \ldots, v_q\}$
and each critical path of $G_0$ passing through 
$(a-v_i,k-1), (a,k), (a+v_i, k+1)$, we reattach one
endpoint of $(a - v_i,k-1)$ to the $i$th input port of
$G_{I, (a,k)}$ and reattach one endpoint of $(a +v_i, k+1)$
to the $i$th output port of $G_{I, (a,k)}$.  
Let $G^*$ be the result of this process.

\paragraph*{The Edge Subdivision Step.}
Every edge in $G^*$ that was inherited from $G_0$ 
(i.e., not inside any copy of $G_I$)
is subdivided into a path of $L/2$ edges.
The outcome of this process is $G$.

\medskip
Observe that for every critical pair $(x, y)$ from $G_0$, 
there is a unique shortest path between $x$ and $y$ in 
$G$ of length $\frac12 L^2 + (L-1)Lc$, where $\frac12 L^2$ edges come
from the subdivision step and the remaining ones come from 
copies of $G_I$.
Moreover, any two unique shortest paths are edge-disjoint.

\begin{lemma}\label{lemma:better-spanner-edges-lowerbound}
Every spanner of $G$ with additive stretch $+(L-2)$ contains at 
least \\
$\left(\frac12 L^2+cL\left(\frac{L-1}{2}\right)\right)|P|$ edges.
\end{lemma}

\begin{proof}
Suppose there exists a spanner $H$ of $G$ with additive stretch $+(L-2)$ but has strictly less than $\left(\frac12 L^2+cL\left(\frac{L-1}{2}\right)\right)|P|$ edges.
By the pigeonhole principle there must exist a critical 
pair $(x, y)\in P$ with unique shortest path $P_{(x, y)}$
that is in one of the following two cases:
(1) $H$ is missing an edge in $P_{(x,y)}$ introduced in 
the edge subdivision step, 
or 
(2) $H$ is missing at least one critical edge from $P_{(x,y)}$
in at least half ($(L-1)/2$) of the 
copies of $G_I$ along $P_{(x,y)}$.

Let $P'_{(x, y)}$ be a 
shortest path connecting $x$ and $y$ in $H$. 
If (1) holds,
then $P'_{(x,y)}$ traverses at least two more subdivided edges
than $P_{(x,y)}$ and at least the same number of copies of $G_I$,
hence $|P'_{(x, y)}| \ge |P_{(x, y)}| + L$, a contradiction.
If (2) holds, then for every inner graph that has a missing edge on $P_{(x, y)}$, $P'_{(x, y)}$ traverses at least two more edges.
Since there are at least $(L-1)/2$ such inner graphs, 
$|P'_{(x, y)}| \ge |P_{(x, y)}| + (L-1)$.  In either case, $P'_{(x, y)}$ has at least $+(L-1)$ additive stretch so $H$ cannot be a $+(L-2)$ spanner.
\end{proof}

\begin{proof}[Proof of Theorem~\ref{theorem:n-spanner-harder}.]
Given the density parameter $c_0$, we will choose a larger 
parameter $c=\Theta(c_0)$ (defined precisely below)
and construct the inner graph $G_I$ (Lemma~\ref{lemma:inner-graph-construction}) 
with at most $qL$ vertices,
at least $cqL$ edges, and $q$ critical pairs, 
with $q = \Theta(L^2c^6)$.
Once $q,L$ are fixed, we construct the outer graph $G_0$
using Lemma~\ref{lemma:spanner-outer-base-graph}.
After the replacement and subdivision steps, $G$
has $|V| = \Theta(L^5q^3)$ vertices and $|P| = \Theta(L^3q^3)$ critical pairs.  This implies that there is some absolute constant
$\lambda>1$ such that 
$\lambda L\left(\frac{L-1}{2} \right)|P| \ge |V|$.

Now, by Lemma~\ref{lemma:better-spanner-edges-lowerbound}, any spanner of $G$ with at most $cL\left(\frac{L-1}{2} \right)|P|$ edges has additive stretch at least $+(L-1)$. 
We choose $c=\lambda c_0$, 
so $cL\left(\frac{L-1}{2} \right)|P| \ge c_0|V|$.
Therefore, any spanner of $G$ with at most $c_0|V|$ edges 
has additive stretch at least $+(L-1)$.
Since $q=\Theta(L^2c^6)=\Theta(L^2c_0^6)$, 
it follows that $|V|=\Theta(L^{11}c_0^{18})$. 
Thus, we conclude that any spanner of $G$ with $c_0n$
edges has additive stretch $+\Omega(|V|^{1/11}c_0^{-18/11})$.
\end{proof}

\begin{remark}
It follows from Theorem~\ref{theorem:n-spanner-harder}
that any $\Theta(n^{1+\epsilon})$-size spanner has 
$+\Omega(n^{\frac{1}{11}-\frac{18}{11}\epsilon})$ stretch, which is only
better than the $+\Omega(n^{\frac{1}{13} - \frac{5}{13}\epsilon})$ bound of Corollary~\ref{cor:spanner-lb} when $\epsilon < 2/181$
is quite small.
\end{remark}

\begin{remark}
The construction from Theorem~\ref{theorem:n-spanner-harder} cannot be easily traslated into an emulator lower bound. The reason is that the number of critical pairs is always sublinear in the number of vertices.
A distance preserving emulator of linear size always exists in this type of construction.
\end{remark}

\subsection{\texorpdfstring{$O(n)$}{O(n)}-Size Emulators}\label{subsection:n-emulators}

The difference between emulators and spanners is that emulators
can use weighted edges not present in the original graph.
In this section, our lower bound graph, $G$, is constructed exactly as 
in Section~\ref{subsection:n-spanners}, 
but with different numerical parameters. 

\begin{theorem}\label{theorem:n-emulator}
There exists an undirected graph $G$ with $n$ vertices such that 
any emulator with $O(n)$ edges has $+\Omega(n^{1/18})$ additive stretch.
\end{theorem}

Before proving Theorem~\ref{theorem:n-emulator} we first argue
that the size of low-stretch emulators is tied to the number of critical
pairs $|P|$ for $G$.

\begin{lemma}\label{lemma:emulator-lemma}
Every emulator for $G$ with additive stretch $+(2D-1)$
requires at least $|P|/2$ edges.
\end{lemma}

\begin{proof}
Let $H$ be an emulator with additive stretch $+(2D-1)$.
Without loss of generality, we may assume that any $(u,v) \in E(H)$
has weight precisely $\dist_G(u,v)$.  (It is not allowed to be smaller, 
and it is unwise to make it larger.)  We proceed to convert $H$ into 
a \underline{spanner} $H'$ that has the same stretch $+(2D-1)$ on all pairs in $P$,
then apply Lemma~\ref{lemma:spanner-construction-lowerbound}.

Initially $H'$ is empty.
Consider each $(x,y)\in P$ one at a time.  
Let $P_{(x,y)}$ be the shortest path in $H$ and 
$P'_{(x,y)}$ be the corresponding path in $G$.
Include the entire path $P'_{(x,y)}$ in $H'$.
After this process is complete, for any $(x,y)\in P$, 
$\dist_{H'}(x,y) = \dist_H(x,y)$, and $H'$ is a spanner with at 
most $n+2D|H|$ edges.  In particular, 
it has at most $2D|H|$ clique edges since each weighted 
edge in some $P_{(x,y)}$ contributes at most $2D$ clique
edges to $H'$.  By Lemma~\ref{lemma:spanner-construction-lowerbound},
the number of clique edges in $H'$ is at least $D|P|$, hence $|H| \ge |P|/2$.
\end{proof}

\begin{proof}[Proof of Theorem~\ref{theorem:n-emulator}]
In order to get $|P|=\Omega(n)$, it suffices to set $\delta_1\ge \delta_2\ge D^2$.

Now, we have 
\begin{align*}
    n &= \Theta\left(R_1^{d_1}R_2^{d_2}D^2\delta_1\right)\\
    &=\Omega\left((r_1D)^{d_1}(r_2D)^{d_2}D^4\right) \tag{$\delta_1\ge \delta_2\ge D^2$}\\
    &=\Omega\left(\left(\delta_1^{\frac{d_1+1}{d_1(d_1-1)}}D\right)^{d_1}
    \left(\delta_2^{\frac{d_2+1}{d_2(d_2-1)}}D\right)^{d_2} D^4\right) \tag{by definition of $\delta_1$ and $\delta_2$}\\
    &=\Omega\left(D^{2\frac{d_1+1}{d_1-1}+d_1+2\frac{d_2+1}{d_2-1}+d_2+4}\right) \tag{$\delta_1\ge \delta_2\ge D^2$}
\end{align*}

The exponent is minimized when $d_1=d_2=3$. This implies $n=\Omega(D^{18})$ and hence $D=O(n^{1/18})$. 
These parameters can be achieved asymptotically by setting
$D=\Theta(n^{1/18})$, $\delta_1=\delta_2=D^2$, 
$r=\Theta(n^{2/27})$, and $R=\Theta(n^{7/54})$.
\end{proof}

\begin{corollary}
Fix an $\epsilon\in[0,1/3)$ and let $d$ be such that 
$\epsilon\in \left[0, \frac{d-1}{3d+1}\right]$. 
There exists a graph $G$ with $n$ vertices such that any emulator $H$ with 
$O(n^{1+\epsilon})$ edges has additive stretch $+\Omega\left(n^{\left(1-\frac{3d+1}{d-1}\epsilon\right)/\left(4+2d+4\frac{d+1}{d-1}\right)}\right)$. In particular, by setting $d=3$ the additive stretch lowerbound becomes $\Omega(n^{\frac{1}{18}-\frac{5}{18}\epsilon})$.
\end{corollary}

Using the same proof technique as in~\cite{abboud_bodwin_2017,AbboudBP17},
it is possible to extend our emulator lower bound to \emph{any} compressed
representation of graphs using $\tilde{O}(n)$ bits.

\begin{theorem}\label{theorem:n-oracle}
Consider any mapping from $n$-vertex graphs to 
$\tilde{O}(n)$-length bitstrings.  Any algorithm
for reconstructing an approximation of $\dist_G$,
given the bitstring encoding of $G$, must have 
additive error $+\tilde{\Omega}(n^{1/18})$.
\end{theorem}

\begin{proof}
For each subset $T\subseteq P$ construct the graph 
$G_T$ by removing all clique edges from $G$ 
that are on the critical paths of pairs in $T$. 
Because all clique edges are missing, for all $(x, y)\in T$ 
we have $d_{G_T}(x, y) \ge d_G(x, y) + 2D$. On the other hand,
for all $(x, y)\notin T$, $d_{G_T}(x, y) = d_G(x, y)$.

There are $2^{|P|}$ such graphs.  If we represent all such graphs with bitstrings of length
$|P|-1$ then by the pigeonhole principle two such graphs $G_T$ and $G_{T'}$ are mapped 
to the same bitstring.  Let $(x,y)$ be any pair in $T\backslash T'$.
Since $\dist_{G_T}(x,y) \ge \dist_{G_{T'}}(x,y) + 2D$, 
the additive stretch of any
such scheme must be at least $2D$.  Alternatively, 
any scheme with stretch $2D-1$
must use bitstrings of length at least length $|P|$.

Now, by setting $d=3$ with $D=\tilde{\Theta}(n^{1/18})$, $r_1=r_2=\tilde{\Theta}(n^{2/27})$ and $R_1=R_2=\tilde{\Theta}(n^{7/54})$, we have 
$|P|=\tilde{\Theta}(n)$.  
Thus any $\tilde{O}(n)$-length encoding must
recover approximate distances with stretch  $+\tilde{\Omega}(n^{1/18})$.
\end{proof}

\section{Conclusion}\label{sect:conclusion}

Our constructions, like~\cite{abboud_bodwin_2017,CoppersmithE06,AbboudBP17,Hesse03}, are based on looking at the convex hulls of integer lattice points 
in $\mathbb{Z}^d$ lying in a ball of some radius.
Whereas Theorems~\ref{theorem:n-emulator} and \ref{theorem:n-oracle} hold for $d=3$,
Theorems~\ref{theorem:n-shortcuts}, \ref{theorem:m-shortcuts}, 
and \ref{theorem:n-spanner} are indifferent between dimensions $d=2$ and $d=3$,
but that is only because \emph{$d$ must be an integer}.

Suppose we engage in a little magical thinking, and imagine that there are 
integer lattices in any \emph{fractional} dimension, and moreover, 
that some analogue of B\'ar\'any and Larman's \cite{Barany98} bound holds
in these lattices.  If such objects existed then we could obtain slightly better
lower bounds.  For example, setting $d=1+\sqrt{2}$ in the proof of Theorem~\ref{theorem:n-shortcuts}, we would conclude that any $O(n)$-size
shortcut set cannot reduce the diameter below $\Omega(n^{1/(3+2\sqrt{2})})$,
which is an improvement over $\Omega(n^{1/6})$ as $3+2\sqrt{2} < 5.83$.

For near-linear size spanners  
and emulators ($n^{1+o(1)}$ edges) 
there are still large gaps 
between the best lower and upper bounds on additive stretch: 
$[n^{1/11},n^{3/7}]$ in the case
of spanners and $[n^{1/18}, n^{1/4}]$ in the case of emulators.  None of
the existing lower or upper bound techniques seem up to 
the task of closing these gaps entirely.

\begin{acks}
Thanks to Greg Bodwin for inspiring the authors which results in improving sparse additive spanner lower bounds in Section~\ref{subsection:improved-n-spanner}.

This work was supported by NSF grants \grantnum{CCF-1514383}{CCF-1514383}, \grantnum{CCF-1637546}{CCF-1637546}, and \grantnum{CCF-1815316}{CCF-1815316}.
\end{acks}

%
\bibliographystyle{ACM-Reference-Format}
\bibliography{sample-base}

\end{document}